\documentclass[twocolumn,pra,nofootinbib]{revtex4}

\usepackage[titletoc,title]{appendix}
\usepackage[utf8]{inputenc}
\usepackage[T1]{fontenc}
\usepackage[english]{babel} 
\usepackage{amsmath}
\usepackage{amsfonts}
\usepackage{amssymb}
\usepackage{amsthm}
\usepackage{euscript}
\usepackage{url}
\usepackage{algorithm,algorithmic}   
\usepackage{tikz}

\newtheorem{theo}{Theorem}

\DeclareMathOperator{\amb}{amb}
\DeclareMathOperator{\bw}{bw}
\DeclareMathOperator{\lazy}{lazy}
\DeclareMathOperator{\target}{target}

\newcommand\ket[1]{|#1\rangle}

\begin{document}

\title{A lazy decoder to reduce decoding hardware requirements for quantum computing}
\title{Hierarchical decoding to reduce hardware requirements for quantum computing}

\makeatletter
\let\@fnsymbol\@arabic
\makeatother

\author{Nicolas Delfosse}
\affiliation{Microsoft Quantum and Microsoft Research, Redmond, WA 98052, USA}
%nidelfos@microsoft.com}

\email{nidelfos@microsoft.com}

\begin{abstract}
Extensive quantum error correction is necessary in order to scale quantum 
hardware to the regime of practical applications. 
As a result, a significant amount of decoding hardware is necessary to process the 
colossal amount of data required to constantly detect and correct errors occurring 
over the millions of physical qubits driving the computation.
The implementation of a recent highly optimized version of Shor's algorithm
to factor a 2,048-bits integer would require more 7 TBit/s of bandwidth for 
the sole purpose of quantum error correction and up to 20,000 decoding units.

To reduce the decoding hardware requirements, we propose a fault-tolerant 
quantum computing architecture based on surface codes with a cheap 
hard-decision decoder, the lazy decoder, combined with a sophisticated 
decoding unit that takes care of complex error configurations. 
Our design drops the decoding hardware requirements by several orders of 
magnitude assuming that good enough qubits are provided.
Given qubits and quantum gates with a physical error rate $p=10^{-4}$, 
the lazy decoder drops both the bandwidth requirements and the number 
of decoding units by a factor 50x.
% naive bw: 2.2 TBit/s
% naive unit: 20,000
% p = 10-4, pTarget = 10-15
% bw = 42,224,000,000 = 42 GBit/s
% unit = 377
Provided very good qubits with error rate $p=10^{-5}$, we obtain a 1,500x reduction
in bandwidth and decoding hardware thanks to the lazy decoder.
% naive bw: 800 GBit/s
% naive unit: 20,000
% p = 10-5, pTarget = 10-15
% bw = 520,000,000 = 520 MBit/s
% unit = 13

Finally, the lazy decoder can be used as a decoder accelerator. Our simulations show a 10x speed-up
of the Union-Find decoder and a 50x speed-up of the Minimum Weight Perfect Matching decoder.
\end{abstract}

\maketitle

Hundreds or thousands of high-quality qubits with an error rate of $10^{-10}$ or lower
are necessary to implement quantum algorithms with industrial applications.
In order to reach such high quality based on current quantum technology, logical qubits 
must be built from a large number of physical qubits and errors accumulating during the
computation must be corrected at regular intervals.

The family of surface codes 
\cite{dennis2002tqm, 
fowler2012surface_code, 
raussenfort2007ft, 
fowler2009highthreshold} 
is the most promising quantum error-correcting scheme to deal with current noise 
levels that barely reach $0.1\%$.
Using a distance $d$ surface code, a logical qubit is encoded into a square grid 
of $d \times d$ data qubits as Fig.~\ref{fig:sc5} shows.
Error correction is based on the measurement 
of $r = d^2 - 1$ {\em syndrome bits}, extracted using syndrome measurement 
circuits implemented on the plaquettes of the qubit grid as shown in Fig.~\ref{fig:sc5}.
The syndrome data is collected by the {\em readout device} and is 
sent to the {\em decoding unit}, which uses this information to detect and 
correct errors.
In order to avoid accumulation of errors during the computation, 
the syndrome is constantly measured, producing $r$ syndrome bits
for each syndrome measurement round.
In the present work, we consider a syndrome measurement time of 
$1 \mu s$, which is the time to implement the four rounds of 
CNOT gates and the final ancilla measurements of the syndrome measurement 
circuit.
Consider as an example the recent RSA factorization algorithm of 
\cite{gidney2019RSA}, 
which relies on distance-27 surface codes, encoding $K \approx 10,000$ 
logical qubits (ignoring distillation qubits).
The implementation of this algorithm requires a bandwidth of 7.3~TBit/s 
and two decoding units per logical qubit, 
%7,280,000,000,000 Bit/s = 7.3 TBit/s
that is 20,000 decoders, assuming independent correction of $X$-type and $Z$-type
Pauli errors.
%\footnote{One could consider 10,000 decoders if we treat simultaneously $X$-type and $Z$-type errors}.
It seems quite challenging to include such a formidable amount of decoding hardware
to ensure fault tolerance in a quantum computer.

\begin{figure}
\centering
\includegraphics[scale=.4]{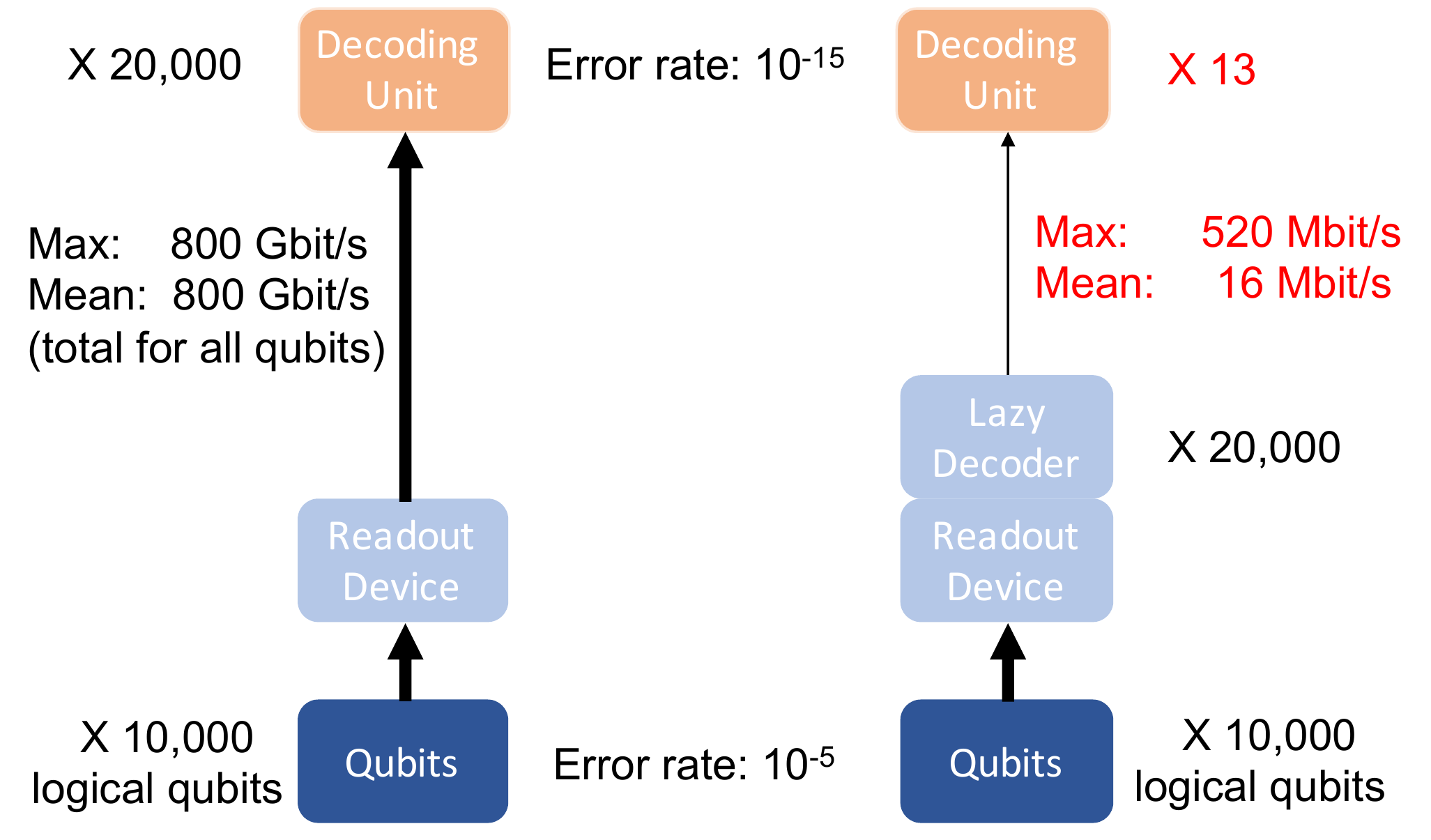}
\caption{Left: Standard design with a readout device sending all syndrome data to the decoding unit.
Two decoding units are used for each logical qubit, one for each type of error.
Right: The readout device is equipped with a lazy decoder unit capable of correcting 
a large number of easy fault configurations, avoiding to transmit syndrome data to the decoding
unit. 
We show the reduction in term bandwidth per logical qubit and number of decoding units obtained 
by introducing the lazy decoder in the case of physical error rate $p=10^{-5}$ 
with a logical error rate $p_L=10^{-15}$.
In this case, including lazy decoding units, saves 
$99.9\%$ of the bandwidth requirement and 
$99.9\%$ of the complex decoding units.
}
\label{fig:design}
\end{figure}

In this work, we propose an alternative to the naive design that allocates 
one decoding unit for each decoding task by introducing a simple hard 
decision decoder, that we refer to as the {\em lazy decoder}.
Fig.~\ref{fig:design} illustrates our design and the saving obtained
for a specific set of parameters.
The lazy decoder can be seen as a pre-decoder
which only attempt to correct easy error configurations.
If no obvious correction exists, it quits and returns a failure mode.
In that case, syndrome data is sent to 
a {\em decoding unit} that hosts a more sophisticated decoder 
achieving a good performance.
Many complex decoding algorithms can play this role 
\cite{dennis2002tqm,
fowler2012autotune,
fowler2012MWPM,
fowler2015parallel_MWPM,
duclos2010RG_decoder,
duclos2013RG_improved,
barrett2010loss_correction_short,
stace2010loss_correction_long,
delfosse2017peeling_decoder,
fowler2013XZ_correlations,
delfosse2014XZ_correlations,
criger2018XZ_correlations_and_degen,
nickerson2019correlation_adaptive_decoder,
dennis2005phd,
harrington2004phd,
wootton2012high_threshold_decoder,
hutter2014MCMC_decoder,
wootton2015simple_decoder,
herold2015FT_SC_decoder,
breuckmann2016local_decoder_2d_4d,
delfosse2017UF_decoder,
tomita2014LUT_decoder,
heim2016optimal_decoder,
torlai2017NN_SC_decoder,
baireuther2018NN_SC_correlations,
krastanov2017DNN_SC_decoder,
varsamopoulos2017NN_SC_decoding,
chamberland2018DNN_small_codes,
breuckmann2018NN_high_dimension_codes,
sweke2018RL_SC_decoder,
baireuther2019NN_CC_circuit_decoder,
ni2018NN_SC_large_d,
andreasson2019DNN_toric_code,
davaasuren2018CNN_SC_decoder,
liu2019NN_LDPC_codes,
varsamopoulos2018NN_SC_decoder_design,
varsamopoulos2019NN_SC_decoder_distributed,
wagner2019NN_training_with_symmetry,
chinni2019NN_CC_decoder,
sheth2019NN_SC_combining_decoder,
colomer2019RL_TC_decoder,
ferris2014TN_QEC,
bravyi2014MLD,
darmawan2018Lineartime_TN_decoding,
tuckett2018TN_biased}

The lazy decoder is designed to be as simple as possible. 
It consists of a single loop over syndrome bits, which makes it 
an ideal candidate for a low-level hardware implementation 
with FPGA or CMOS and it is also easy to parallelize.
We picture the lazy decoder as a hardware unit, as close as
possible to the readout device. 
In the worst case, we need two lazy decoder per logical qubit 
but given the speed of this module, one can expect sharing 
this unit between many of qubits.
We assume that the measurement device (and therefore the lazy decoder) 
are placed in the proximity of the qubits in order to avoid 
long feedback loops increasing the physical qubit clock cycle.

The decoding unit, which is significantly more complex than the lazy decoder, 
may be challenging to implement close to the qubits without introducing 
additional noise to the quantum plane. 
Therefore, we consider a decoding unit placed at further distance 
from the qubits, which leads to latency issues, justifying our focus
on the bandwidth of the readout-decoding unit link. We ignore
the bandwidth between the readout-device and the adjacent lazy
decoder.

Introducing the lazy decoder reduces the bandwidth required to send 
syndrome data from the readout device to the decoding unit because 
in most cases the nearby lazy decoder takes care of the correction.
Moreover, the number of decoding units required is significantly reduced
if one can rely on the lazy decoder a large fraction of the time.
In what follows, we prove that this design leads to a reduction of 
the decoding hardware of several orders of magnitudes for good enough
qubits. 
%In order to benefit from the lazy decoder, one must reach 
error rates below the standard assumption of $10^{-3}$.

\section{The surface code}

%Let us first recall the basic properties of surface codes.
The surface code encodes a logical qubit into a square grid of $d \times d$ 
{\em data qubits} where $d$ is the minimum distance of the code, 
as show in Fig.~\ref{fig:sc5}.
Error correction relies on the syndrome measurement circuits 
represented in Fig.~\ref{fig:sc5}, consuming an additional $d^2 - 1$ qubits.

All plaquettes are measured simultaneously at regular intervals, producing rounds of 
syndrome data for the decoder which identifies errors based on this information. 
Fig.~\ref{fig:sc5} shows the schedule used for the sequence of CNOT 
gates in order to allow for a parallel implementation and to preserve the code distance
despite the propagation of errors by CNOT gates.

We simulate the surface code and the syndrome extraction circuit with a
{\em circuit level noise} that represents imperfections on all qubits, gates, measurements,
and waiting steps, by injecting random Pauli faults between any two steps of the circuit.
A single qubit fault is included after each state preparation or waiting qubit 
with probability $p$. The Pauli fault is chosen uniformly in the set $\{X, Y, Z\}$. 
The outcome of a single qubit measurement is flipped with probability $2p/3$.
The CNOT noise is modeled by a two-qubit Pauli fault injected after the CNOT
with probability $p$.
The fault is selected uniformly between the 15 non-trivial two-qubit Pauli operators
acting on the support of the CNOT gate. 

\begin{figure}
\centering
\includegraphics[scale=.5]{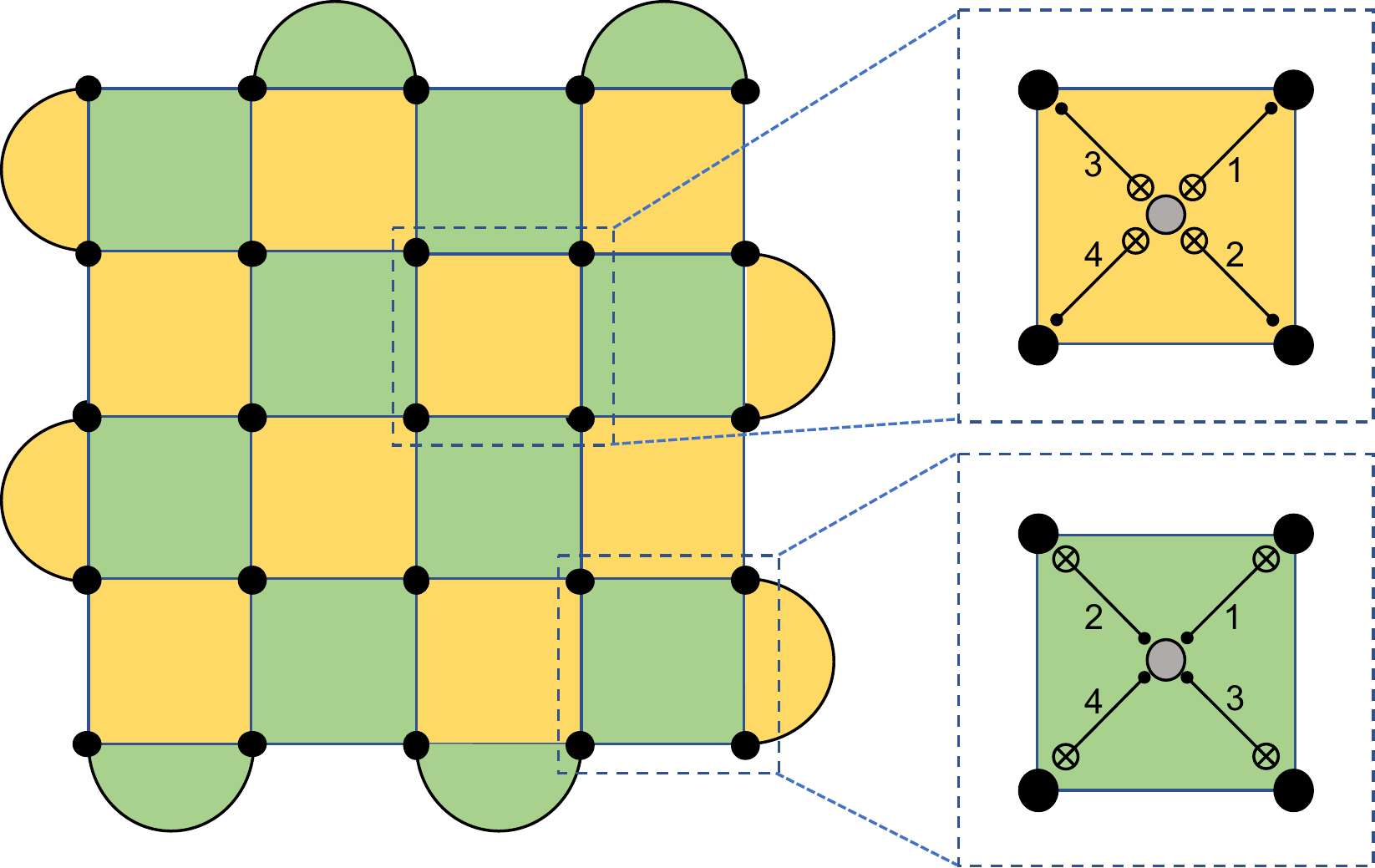}
\caption{Distance five surface code with 25 (black) data qubits encoding 1 logical qubit.
The measurement circuit on green and yellow plaquettes consume a (grey) ancilla qubit 
per plaquette.
A syndrome extraction round is performed with three steps implemented simultaneously over all 
all the plaquettes:
(i) Prepare the ancilla in the state $\ket +$ (green) or $\ket 0$ (yellow).
(ii) For each plaquette apply the four CNOT gates in the order prescribed by their indices,
(iii) Measure the ancilla in the $Z$-basis (green) or in the $X$-basis (yellow).
The final measurement produces a syndrome bit for each plaquette.
Boundary plaquettes implement the restriction of the plaquette measurement 
circuit to two qubits.
}
\label{fig:sc5}
\end{figure}

Equipped with qubits and quantum gates affected by a circuit level noise 
with rate $p$, the surface code encoding provides 
a logical qubit whose error rate drops to \cite{fowler2012surface_code}
\begin{equation} \label{eqn:sc_logical_error_rate}
p_{L}(p, d) = 0.1 (100 p)^{(d+1)/2} \cdot
\end{equation}
where $d$ is the surface code minimum distance.
One can use this heuristic formula to estimate the minimum distance 
$d$ required in order to achieve a given target logical error rate.
Most practical applications necessitate a logical error rate that varies between
$10^{-10}$ and $10^{-15}$, which is out of reach on current hardware
without error correction.

\begin{figure}
\centering
\includegraphics[scale=.3]{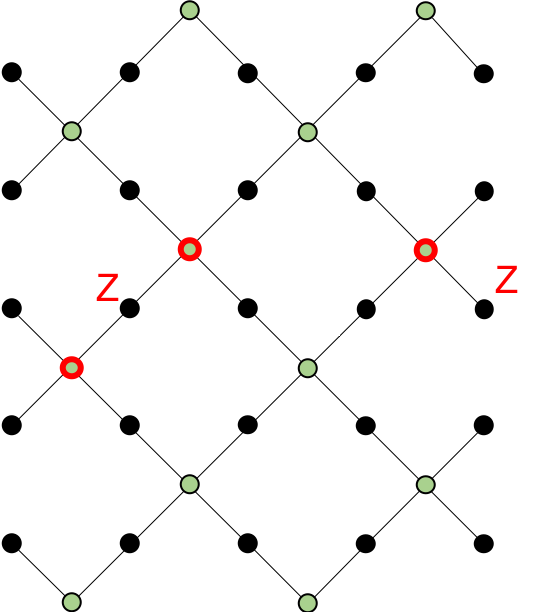}
%\hspace{.5cm}
%\includegraphics[scale=.4]{fig/sc5_error2}
\caption{A slice of the decoding graph obtained by placing a vertex 
in the center of each green plaquette and connecting incident plaquettes
with an edge. A $Z$-fault is detected by either one or two plaquettes as
shown with marked vertices.
}
\label{fig:sc5_error}
\end{figure}

After producing an encoded surface code state, we start measuring syndrome 
data.
Non-trivial syndrome values indicate the presence of a fault.
For simplicity, we focus on $Z$-type faults, detected by 
the measurement of green plaquettes as in Fig~\ref{fig:sc5_error}.
$X$-type faults can be treated similarly.
In what follows, we describe a graph that represents all possible faults
in the syndrome measurement circuit. The whole simulation of the 
syndrome extraction circuit can be implemented based on this graph.

Consider the space-time locations $(x, y, t)$ of syndrome bits, 
where $(x, y)$ is the coordinates of the center of a plaquette and 
$t = 0, 1, 2, \dots$ is the index of the syndrome round.
Let $s(x, y, t) = 0$ or $1$ be the syndrome value in location $(x, y, t)$.
We record the changes of syndrome values, that is 
$\bar s(x, y, t) = s(x, y, t) -  s(x, y, t-1) \pmod 2$, 
and setting $\bar s(x, y, 0) = 0$ for the first round.
A fault in the measurement circuit is detected by a set of 
syndrome locations $(x, y, t)$ where the syndrome value changes, 
{\em i.e.} $s(x, y, t) \neq 0$
A non-trivial fault is detected either in one location $u$ or in a pair of 
locations $u, v$, leading a natural graph structure.

The {\em decoding graph} represents all possible faults in the syndrome 
extraction circuit.
The vertex set of the decoding graph is the set of syndrome locations. 
A half-edge $\{u, -\}$ or an edge $\{u, v\}$ is built from each potential 
fault in the syndrome extraction circuit.
The decoding graph is a 3D cubic lattice with additional diagonal edges.
Fig~\ref{fig:sc5_error} shows a horizontal slice of the decoding graph.
For each edge, we also store the probability of all circuit faults that map onto 
that edge. This edge weight can be processed by the decoder to provide a 
more accurate correction.

A surface code decoder takes as an input a set of consecutive rounds 
of syndrome data given by $\bar s$ and it aims at identifying the 
residual error on the $d^2$ data qubits.
A standard decoding strategy consists in identifying a minimum set of edges 
that matches the observed syndrome $\bar s$.
%A number of decoding algorithms have been designed in order to solve this
%problem efficiently.
%Unfortunately their complexity makes a cold implementation quite challenging.

\section{Lazy decoder}

To simplify, we correct separately $X$-type and $Z$-type faults. 
%We assume that the
%decoding unit accumulates $d$ rounds of syndrome data before providing an estimation of 
%the correction.
Our objective is not to design a good decoder but to identify a set of fault 
configurations that is both very likely and easy to correct. 
The lazy decoder will correct exclusively this subset of easy configurations.
%When the lazy decoder succeeds, not other decoding unit is necessary
%shrinking the decoding hardware needs.

Let us first describe the most naive version of the lazy decoder.
We simply check whether the syndrome is trivial and if so we send no data 
to the decoding unit. Clearly, it is easy to implement with low hardware costs. 
However, this does not help to reduce decoding hardware requirements
since the probability to observed a trivial syndrome is generally too small.
%For a distance-27 surface code with a physical error rate $p=10^{-3}$,
%as in the factorization algorithm cited above, the average bandwidth used is 
%only divided by two with this naive strategy.
One could design a decoder that corrects any single fault or any fault of weight two, 
three and so on, but our numerical experiments show that 
either these sets of faults are still not likely enough for reasonable distances 
or the lazy decoder design becomes just as complex as the design of 
the whole decoding unit, defeating the whole purpose of the lazy decoder.

Algorithm~\ref{algo:lazy_decoder} proposes a satisfying version of the lazy decoder 
for our architecture.
We will prove that, when it succeeds, it returns a minimum set of faults explaining 
the syndrome observed. 
Our basic idea is to correct all configurations that can be corrected locally.
This also guarantees a high potential for parallelism.
If faults are sufficiently separated from each-other in space and time, we can obtain 
such a globally minimal solution from obvious locally optimal decisions. 

The syndrome is given as an input of Algorithm~\ref{algo:lazy_decoder}
as a set of vertices $\bar s \subset V$ in the decoding graph induced by 
$Z$-type faults and the algorithm returns either an estimation 
${\cal E} \subset E$ of the set of edges supporting faults, or a failure mode.
The first block of Algorithm~\ref{algo:lazy_decoder} looks for edges that 
match two neighbors syndrome bits $u, v \in \bf s$. 
Such an edge is locally optimal (it is the minimum number of faults 
explaining two non-trivial syndrome nodes) 
and can be safely added to the correction $\cal E$.
The second block of Algorithm~\ref{algo:lazy_decoder} takes care of 
remaining unmatched syndrome vertices, that come from faults on half-edges.
The half-edge $\{u, -\}$ is a locally optimal choice to explain the non-trivial syndrome 
$\bar s(u) = 1$ only if $u$ has no neighbor $v$ supporting a non-trivial syndrome value.
Otherwise, the choice of $\{u, -\}$ is said to be {\em ambiguous}.
We use the notation $N_v$ for the set of neighbors of a vertex $v$.
In order to guarantee a globally optimal	solution, we count the number of ambiguous
choices $N_{\amb}$ and we return {\bf failure} if at least two ambiguous half-edges are
present in $\cal E$.
Theorem~\ref{theo:lazy_dec_optimality} proves the optimality of our strategy.

\begin{algorithm}
\caption{Lazy decoder}
\label{algo:lazy_decoder}

\begin{algorithmic}[1]
\REQUIRE A syndrome set $\bar s \subset V$.
\ENSURE Either a fault set ${\cal E} \subset E$ such that $\bar s({\cal E}) = \bar s$
or {\bf failure}.
\STATE Set $s' = s$, $\cal E = \emptyset$ and $N_{\amb} = 0$.
\STATE Run over all edges $e = \{u, v\} \in E$ and do:
\STATE	\hspace{.5cm}  If $u \in \bar s'$ and $v \in \bar s'$:
\STATE 		\hspace{1cm} Add $e$ to $\cal E$ and remove $u$ and $v$ from $\bar s'$.
\STATE Run over all half-edges $e = \{u, -\} \in E$ and do:
\STATE 	\hspace{.5cm} If $u \in \bar s'$ do:
\STATE 		\hspace{1cm} Add $e$ to $\cal E$ and remove $u$ from $\bar s'$.
\STATE		\hspace{1cm} If $N_v \cap \bar s \neq \emptyset$, increment $N_{\amb}$
\STATE			\hspace{1.5cm} If $N_{\amb} > 1$ return {\bf failure}.
\STATE If $s' \neq \emptyset$ return {\bf failure}
\STATE Return $\cal E$.
\end{algorithmic}
\end{algorithm}

\begin{theo} \label{theo:lazy_dec_optimality}
If the lazy decoder succeeds it returns a minimum set of faults for the syndrome $\bar s$.
\end{theo}

\begin{proof}
If $\bar s({\cal E}) = \bar s$, the set $\cal E$ contains a set of paths that 
connects vertices of $\bar s$ either by pairs or to the boundary.
Naively, we have 
$
|{\cal E}| \geq |\bar s|/2
$,  
with equality if and only if $\cal E$ pairs each vertex of $\bar s$ with one of its neighbors.
%In the case of an isolated vertex in the bulk of the graph, Algorithm~\ref{algo:lazy_decoder}
%returns {\bf failure}.

Let $\partial \bar s$ be the set of vertices $v$ of $\bar s$ indicent to
a half-edge $\{v, -\}$.
Consider the subset $\partial \bar s^* \subset \partial \bar s$ of vertices 
$v$ that have no neighbor in $\bar s$ (that is $N_v \cap \bar s = \emptyset$). 
For an arbitrary fault set ${\cal E} \subset E$, any vertex of $\partial \bar s^*$ is either part of a 
half-edge or it is connected to a vertex at distance $\geq 2$, leading to
the bound
\begin{equation} \label{eqn:proof_theo_lower_bound}
|{\cal E}| \geq ( |\bar s| - |\partial \bar s^*| ) /2 + |\partial \bar s^*| \cdot
\end{equation}
This equation is satisfied for all fault sets $\cal E$ with syndrome $\bar s$.

Consider now the fault set $\cal E$ produced by Algorithm~\ref{algo:lazy_decoder}
in case of success.
The first block finds edges that match bulk vertices 
$v \in \bar s \backslash \partial \bar s$ to a neighbor.
Vertices of $\partial \bar s^*$ are linked to a boundary by a half-edge in $\cal E$.
Finally, the vertices $v \in \partial \bar s \backslash \partial \bar s^*$
are all matched to a neighboring vertex except at most one 
(because in case of success we have $N_{\amb} \leq 1$).

This proves that the set $\cal E$, returned by Algorithm~\ref{algo:lazy_decoder}
in case of success, satisfies
\begin{equation} \label{eqn:proof_theo_upper_bound}
|{\cal E}| \leq ( |\bar s| - |\partial \bar s^*| ) /2 + |\partial \bar s^*| + 1/2 \cdot
\end{equation}
Together with the lower bound \eqref{eqn:proof_theo_lower_bound},
this demonstrates that the size of $\cal E$ is minimum.
\end{proof}

One can perform the lazy decoding on the fly while reading the syndrome rounds.
It is enough to store three consecutive rounds of syndrome values to apply the
lazy decoder. 
When a failure of the lazy decoder is detected, we start sending syndrome information 
to the decoding unit which accumulates $d$ rounds of data to provide
a correction. This leads to an asynchronous decoding between different logical blocks,
that can be advantageous to share decoding hardware between logical qubits 
but that could induce stalling in the layout of logical operations. 
We do not explore the consequences of this asynchronous decoding in the current work.
The locality of Algorithm~\ref{algo:lazy_decoder} suggests an easy parallel 
implementation. Only the value $N_{\amb}$ is a global data.

\section{Bandwidth reduction}

Without the lazy decoder, the bandwidth used per logical qubit is 
$
\bw(d) = (d^2 - 1) / \tau
$
bits, where $d$ is the code distance and $\tau$ is the time required 
per syndrome extraction round in seconds.
All the numerical results of this article are obtained assuming $\tau = 1 \mu s$.

The readout-decoding unit bandwidth used for a logical qubit drops 
to zero while the lazy decoder succeeds. 
This induces a significant reduction of the average bandwidth 
used per logical qubit, as we can see in Fig.~\ref{fig:bandwidth_d5_d15_d25}.
With physical error rate $p=10^{-4}$, the average bandwidth saving 
varies between 1 order of magnitude for the distance-35 surface code to more
than 3 orders of magnitude for distance $d=5$.

\begin{figure}
\centering
\includegraphics[scale=.5]{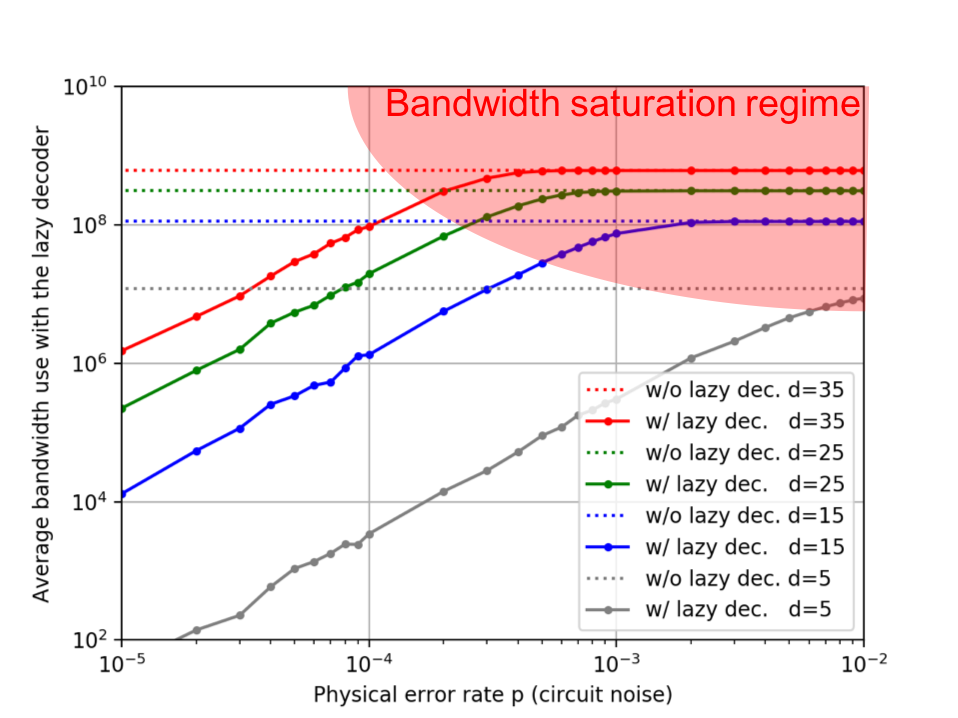}
\caption{Average bandwidth used per logical qubit in Bit/s 
with the lazy decoder for surface codes with distance $d=5, 15, 25, 35$.
Dashed (continuous) lines represent the average bandwidth used per logical qubit without (with)
lazy decoder.
The lazy decoder does not help in reducing the bandwidth in the bandwidth saturation
regime.
}
\label{fig:bandwidth_d5_d15_d25}
\end{figure}

We observed a phenomenon of {\em bandwidth saturation}
which occurs when using a large-distance code with a qubit quality 
that is not far enough below the threshold, {\em e.g.}
$d \geq 15$ with $p=10^{-3}$.
In this regime, the lazy decoder almost constantly fails and 
we do not observe any reduction of the bandwidth. Then, it
may be preferable to remove the lazy decoder to avoid
hurting the decoder's performance by additional latency.
This suggests that it is necessary to keep improving qubit quality 
far below the surface code threshold in order to scale up quantum 
hardware and its classical control to reach the regime of practical 
applications.

\section{Bandwidth requirements}

We observed a neat reduction of the average bandwidth use
using the lazy decoder. 
However, the bandwidth utilization varies with time and
the system often requires much more bandwidth than the 
average use. The required bandwidth and the number of decoding
unit needed depends on the maximum number of failures of the 
lazy decoder over the $K$ logical qubits of the quantum computer.
%Consider a quantum computer with $K$ logical qubits, each
%encoded with a distance-$d$ surface code, using quantum 
%hardware with physical error rate $p$.

To simplify, we consider a single communication channel connecting the 
readout devices of all logical qubits to the decoding units.
A {\em bandwidth failure} occurs if at given point in time the 
bandwidth needs for the whole system surpass the bandwidth 
of the readout-decoder channel.

%Similarly, a {\em decoding unit failure} occurs if we do not have
%enough to decoding units to allocate one decoder for each 
%logical qubit where the lazy decoder fails.

The {\em bandwidth required} is defined to be the minimum bandwidth 
such that the probability of bandwidth failure is smaller than $p_{L}$,
which guarantees that the bandwidth bottleneck is not the dominant 
source of system failure (as suggested in \cite{das2020:UF_architecture}).
To obtain the bandwidth required, consider the failure probability
$p_{\bf fail} = p_{\bf fail}(p, d)$ for the lazy decoder
over $d$ consecutive rounds of syndrome measurement for a single 
logical qubit.
We assume that the noise on different logical qubits is independent,
so that the probability of at least $m$ failures of the lazy decoder over the 
$K$ logical qubits is given by 
$
\binom {2K} m p_{\bf fail}^m.
$
The bandwidth required for the whole system of $K$ logical qubits 
is given by 
\begin{equation} \label{eqn:bw_req}
\bw_{\lazy}(p, d, K) =  M(p, d, K) \bw(d)
\end{equation}
where $M = M(p, d, K)$ is the smallest integer such that  
\begin{equation} \label{eqn:n_lazy_fail}
\binom {2K} {M + 1 } p_{\bf fail}(p, d)^{ M + 1 } < p_L(p, d) 
\end{equation}
which ensures that a bandwidth failure occurs with probability
at most $p_{L}$.
It may be challenging to evaluate numerically $M(p, d, K)$ based in Eq.~\eqref{eqn:n_lazy_fail}
for large values of $K$. The numerical results presented below 
rely on Chernoff bound to derive an upper bound on $M(p, d, K)$.

\begin{figure}
\centering
\includegraphics[scale=.6]{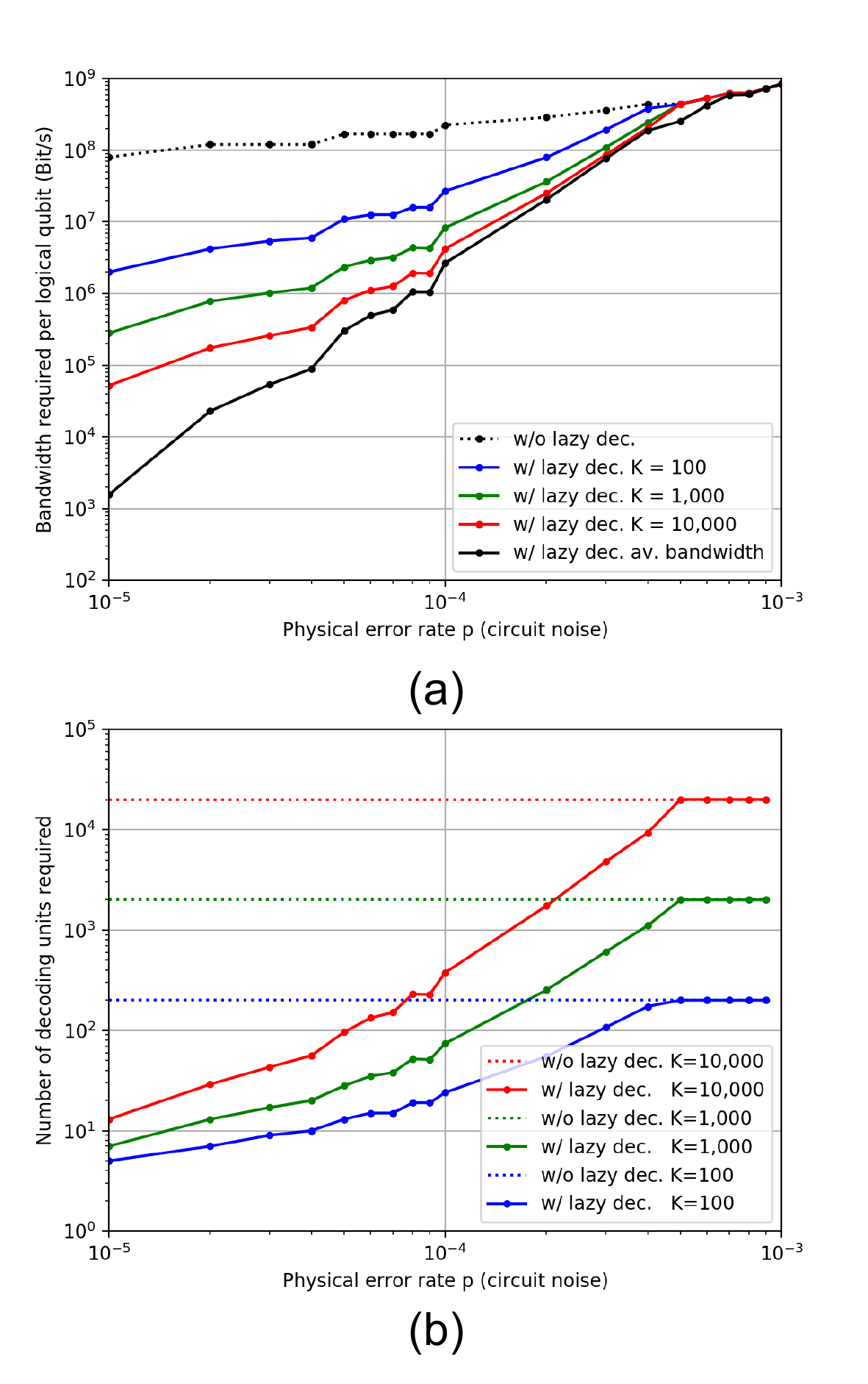}
\caption{
(a)
Bandwidth required per logical qubit in Bit/s for different 
number of logical qubits $K$ to reach a target logical 
error rate of $10^{-15}$.
(b)
Number of decoding units required for a system of $K$ 
logical qubits in order to reach a target logical error rate 
of $10^{-15}$. 
}
\label{fig:bandwidth_req}
\end{figure}

Fig.~\ref{fig:bandwidth_req} shows the bandwidth required to
reach a target logical error rate $p_{\target} = 10^{-15}$.
Given $p_{\target}$ and the physical error rate $p$ of the device, 
we first pick the smallest minimum distance $d$ that ensures
$p_{L}(p, d) < p_{\target}$ using Eq.~\ref{eqn:sc_logical_error_rate}.
The minimum distance varies discretely with $p$, inducing brutal jumps 
in the system requirements.
Once the distance is fixed, we estimate the lazy decoder failure 
probability $p_{\bf fail}(p, d)$ by a Monte-Carlo simulation, from
which we derive the value of $\bw_{\lazy}(p, d, K)$ based on 
Eq.~\eqref{eqn:bw_req}.
A better distribution of resources is achieved for a system that
contains many logical qubits, dropping the bandwidth required per logical qubits 
closer to the average use.
The bandwidth saturation appears again in the regime $p = 10^{-3}$
where we require almost 1GBits/s per logical qubit. 
For error rates $p \geq 6.10^{-4}$, we observe no saving for bandwidth 
requirements.

\section{Decoding hardware requirements}

In addition to a substantial bandwidth reduction, the lazy decoder 
induces savings in the decoding hardware. 
Indeed, the value $M(p, d, K)$ introduced in Eq.~\eqref{eqn:bw_req}
is the largest number of decoding tasks to perform simultaneously 
over the whole system of $K$ logical qubits.
Instead of allocating one decoding unit for each logical qubit, one
can share $M(p, d, K)$ decoding units without notably affecting the 
failure rate of the quantum computer.
Fig.~\ref{fig:bandwidth_req}, shows the saving in term of number 
of decoding units.
In order to reach a target error rate of $10^{-15}$ 
with a system of $K = 10,000$ logical qubits with physical error rate 
$p=10^{-4}$ (resp. $10^{-5}$)
a naive design uses $2K = 20,000$ decoding units while only 377 units 
(resp. 13) are sufficient with the lazy decoder, 
saving $98\%$ (resp. $99.9\%$) 
of the decoding hardware.
Table~\ref{tab:hw_req} shows the saving and the hardware requirements
for different target noise rate, qubit quality and system size.
Again, the saturation in the regime $p=10^{-3}$ limits the saving. 
We need better qubits in order to scale up quantum computers to
the massive size required for practical applications.

\begin{table}[htb]
\begin{center}
\begin{footnotesize}
\caption{Decoding Hardware Requirements with Lazy Decoder for different system sizes.
We indicate the fraction of decoding hardware saved by the Lazy decoder.}
\setlength{\tabcolsep}{1.2mm} 
\renewcommand{\arraystretch}{1.0}
\label{tab:hw_req}
\begin{tabular}{ |l|l|l|l| } 
\hline
	\multicolumn{4}{|c|}{ $p_{target} = 10^{-15}$ } \\
\hline
\hline	
	$p$
	& $K = 100$
	& $K = 1,000$
	& $K = 10,000$
	\\
\hline
 	& $d = 29$
	& $d = 29$
	& $d = 29$
	\\
	$10^{-3}$
 	& 84 GBit/s
 	& 840 GBit/s
 	& 8.4 TBit/s
	\\
 	& 200 dec. units
 	& 2,000 dec. units
 	& 20,000 dec. units
	\\
 	& save $0\%$
	& save $0\%$
	& save $0\%$
	\\
\hline
 	& $d = 15$
	& $d = 15$
	& $d = 15$
	\\
	$10^{-4}$
 	& 2,7 GBit/s 
	& 8,3 GBit/s 
	& 42 GBit/s
	\\
 	& 24 dec. units
 	& 74  dec. units
 	& 377 dec. units
	\\
 	& save $88\%$
	& save $96\%$
	& save $98\%$
	\\
\hline
 	& $d = 9$
	& $d = 9$
	& $d = 9$
	\\
	$10^{-5}$
 	& 200 MBit/s
	& 280 MBit/s
	& 520 MBit/s
	\\
 	& 5 dec. units
	& 7 dec. units 
	& 13 dec. units
	\\
 	& save $97.5\%$
	& save $99.7\%$
	& save $99.9\%$
	\\
\hline
\hline
	\multicolumn{4}{|c|}{ $p_{target} = 10^{-12}$ } \\
\hline
\hline	
	$p$
	& $K = 100$
	& $K = 1,000$
	& $K = 10,000$
	\\
\hline
 	& $d = 23$
	& $d = 23$
	& $d = 23$
	\\
	$10^{-3}$
 	& 53 GBit/s
 	& 530 GBit/s
 	& 5.3 TBit/s
	\\
 	& 200 dec. units
 	& 2,000 dec. units
 	& 20,000 dec. units
	\\
 	& save $0\%$
	& save $0\%$
	& save $0\%$
	\\
\hline
 	& $d = 11$
	& $d = 11$
	& $d = 11$
	\\
	$10^{-4}$
 	& 900 MBit/s 
	& 2,4 GBit/s 
	& 10.5 GBit/s
	\\
 	& 15 dec. units
 	& 40  dec. units
 	& 175 dec. units
	\\
 	& save $93\%$
	& save $98\%$
	& save $99\%$
	\\
\hline
 	& $d = 7$
	& $d = 7$
	& $d = 7$
	\\
	$10^{-5}$
 	& 96 MBit/s
	& 144 MBit/s
	& 216 MBit/s
	\\
 	& 4 dec. units
	& 6 dec. units 
	& 9 dec. units
	\\
 	& save $98\%$
	& save $99.7\%$
	& save $99.96\%$
	\\
\hline
\hline
	\multicolumn{4}{|c|}{ $p_{target} = 10^{-9}$ } \\
\hline
\hline	
	$p$
	& $K = 100$
	& $K = 1,000$
	& $K = 10,000$
	\\
\hline
 	& $d = 17$
	& $d = 17$
	& $d = 17$
	\\
	$10^{-3}$
 	& 29 GBit/s
 	& 290 GBit/s
 	& 2.9 TBit/s
	\\
 	& 200 dec. units
 	& 2,000 dec. units
 	& 20,000 dec. units
	\\
 	& save $0\%$
	& save $0\%$
	& save $0\%$
	\\
\hline
 	& $d = 7$
	& $d = 7$
	& $d = 7$
	\\
	$10^{-4}$
 	& 168 MBit/s 
	& 360 MBit/s 
	& 1.2 GBit/s
	\\
 	& 7 dec. units
 	& 15  dec. units
 	& 51 dec. units
	\\
 	& save $97\%$
	& save $99.3\%$
	& save $99.8\%$
	\\
\hline
 	& $d = 5$
	& $d = 5$
	& $d = 5$
	\\
	$10^{-5}$
 	& 24 MBit/s
	& 36 MBit/s
	& 60 MBit/s
	\\
 	& 2 dec. units
	& 3 dec. units 
	& 5 dec. units
	\\
 	& save $99\%$
	& save $99.9\%$
	& save $99.98\%$
	\\
\hline
\end{tabular}
\end{footnotesize}
\end{center}
\end{table}

\section{The lazy decoder as a decoder accelerator}

The lazy decoder can be considered as a (hardware or software) decoder accelerator.
It speeds up any decoding algorithm without significantly degrading the 
correction capability.
Fig.~\ref{fig:lazy_accelerator} shows the average execution time for our implementation
in C of two standard decoding algorithms with and without lazy pre-decoding.
The speed-up reaches a factor 10x for the Union-Find (UF) decoder \cite{delfosse2017UF_decoder},
which is already one of the fastest decoding algorithms and we obtain a 50x acceleration of the 
Minimum Weight Perfect Matching (MWPM) decoder \cite{dennis2002tqm}.
Note that both combinations Lazy + UF and Lazy + MWPM achieve a similar average runtime,
although the worst-case execution time, which is a central parameter in the design of a 
decoding architecture \cite{das2020:UF_architecture}, is substantially larger
for the MWPM. 

We also confirmed numerically that the lazy decoder does not deteriorate 
the performance of the MWPM decoder and the UF decoder 
as Theorem~\ref{theo:lazy_dec_optimality} suggests. 
On the contrary, the lazy decoder provides a slight improvement of the correction capacity 
of the UF decoder. This is because these two algorithms perform well on different 
types of fault configurations. 
The work of Seth et al. \cite{sheth2019:decoder_combination} explores further the
idea of combining different decoding strategies.

\begin{figure}
\centering
\includegraphics[scale=.5]{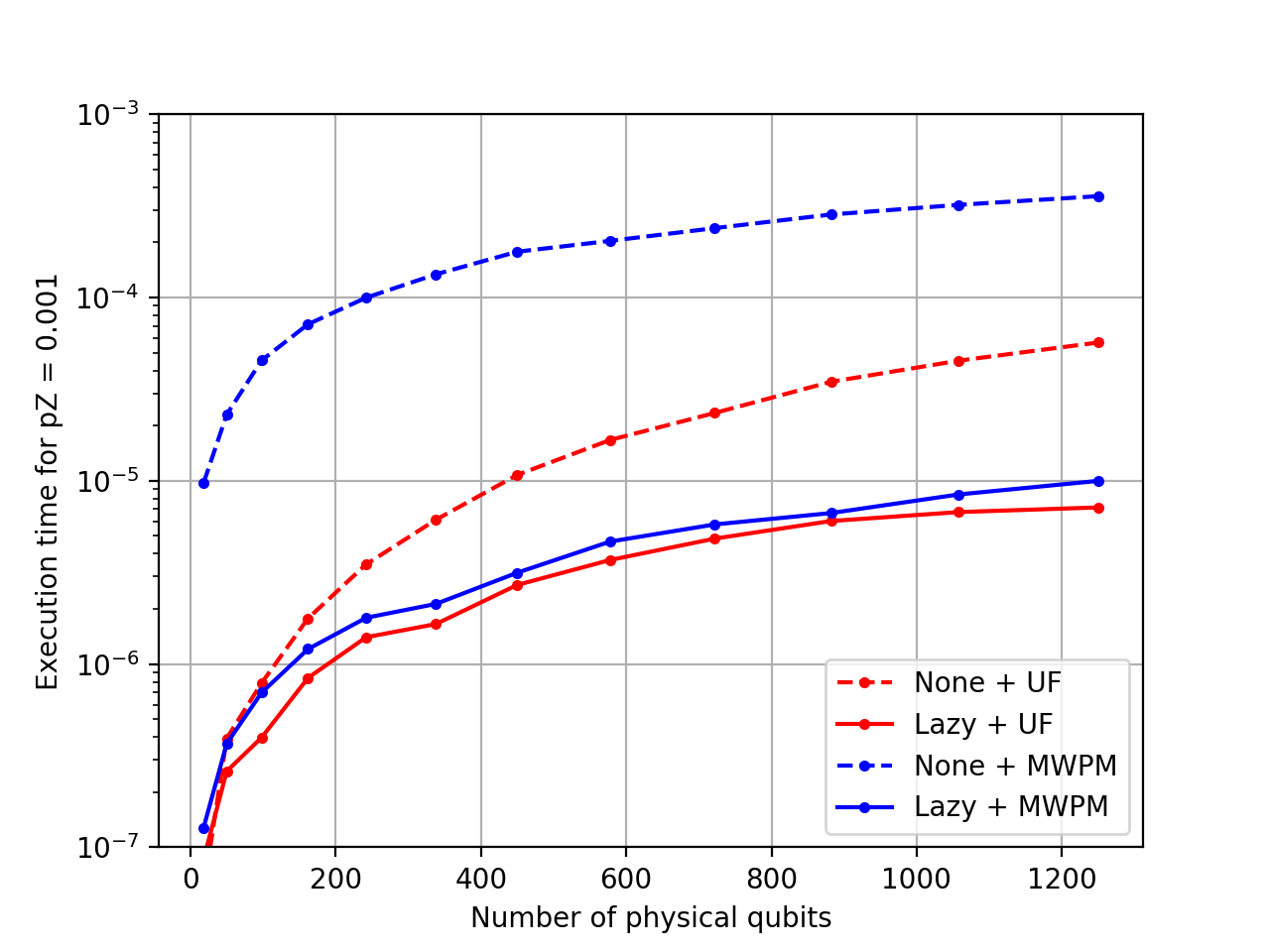}
\caption{
Execution time in seconds of the MWPM decoder and the UF decoder 
with and without lazy decoder. The runtime is estimated over $10^6$ trials,
over the 2D toric code, assuming perfect measurements and an error rate of $10^{-3}$.
We use an implementation in C of these algorithms executed on a MacBook Pro 2013
with a single thread processor 2,4 GHz Intel Core i5.
We observe a 10x speed-up of the UF decoder and 50x for the MWPM decoder.
}
\label{fig:lazy_accelerator}
\end{figure}

\medskip
{\em Conclusion --}
Error correction is a major bottleneck in fault-tolerant quantum computing
which leads to a huge overhead in the implementation of quantum algorithms 
\cite{reiher2017:nitrogen, campbell2019:CSP, gidney2019RSA}.
In this article, we designed a simple decoder that can be used in combination with a
more complex decoding unit to correct errors simultaneously on many logical qubits 
with a minimum decoding hardware. 

In future work, we plan to explore the impact of serialization latency on the decoder's 
performance. 

Although this work focuses on surface codes, the general principle 
of this design can be adapted to any quantum error correction code
if we can identify a set of easy error configuration that is likely enough.

The lazy decoder applies to any type of surface code, including codes defined on non-trivial
topology \cite{zemor2009:hom_codes, breuckmann2017hyperbolic}.
The lazy decoder can be directly applied to color codes \cite{bombin2006color_codes}
using for instance the projection decoder \cite{delfosse2014projection, kubica2019restriction}.

Beyond topological codes, one can adapt the lazy decoder to quantum LDPC 
codes 
\cite{mackay2004:QLDPC, tillich2013QLDPC, gottesman2014QLDPC, fawzi2018QLDPC, campbell2019:single_shot, breuckmann2020:LDPC_single_shot}.
The sparsity of the Tanner graph guarantees the success of our local strategy
for low enough physical error rate.

The basic idea of using a pre-decoder dedicated to the correction of simple configurations
is also central in the design of a flash memory controller where a hard-decision belief 
propagation (BP) decoder is used as a pre-decoder 
and, in case of failure, multiple levels of soft-decision BP are performed \cite{cai2017:flash_decoding}.
However, the noise rate of flash cells is far more favorable than in quantum hardware,
allowing for using a single decoding unit to correct many encoded blocks in flash memory.
Note that the execution time current flash BP decoders are far too long for the quantum setting
if we suppose that the decoding must be implemented in $d \mu s$.
($80 \mu s$ for hard decision decoder + $80 \mu s$ per level of soft-BP) \cite{cai2017:flash_decoding}.

The BP decoder provides a hierarchy of decoding algorithms 
with growing complexity as a function of the number of propagation levels. 
This flexibility allows for adjusting the number of levels in order 
to maximize the success probability of the decoder according the decoding available time.
The Union-Find decoder \cite{delfosse2017UF_decoder} offers the same advantage 
by tuning the number of growth rounds.

In the future, it would be interesting to explore further the hardware implementation 
of the Lazy decoder following the approach of \cite{das2020:UF_architecture} 
and to fabricate an FPGA or ASIC prototype in order to obtain a better insight 
on practical applications of the lazy decoder.

\medskip
{\em Acknowledgement -- The author would like to thank Krysta Svore, Michael Beverland, Jeongwan Haah, Adam Paetznick, Chris Pattison, Poulami Das, Alan Geller, Matthias Troyer, Helmut Katzgraber and Dave Wecker for insightful discussions and Poulami Das for her comments on a preliminary version of this article.
}

%\bibliographystyle{apsrev}
%\bibliography{bib_lazy}

\end{document}